\documentclass[a4paper, 11pt, headings = small, abstract,numbers=endperiod]{scrartcl}

\usepackage{amscd,amsmath,amsthm,array,bbm,hhline,mathrsfs,dsfont,changes, booktabs,fancybox,calc,textcomp,xcolor,graphicx,bbm,xspace,nicefrac,stmaryrd,url,arcs,listings,nameref}
\usepackage{enumitem}

\usepackage[title]{appendix}
\usepackage[semibold]{libertine}
\usepackage[upint]{libertinust1math}
\usepackage[scr=boondoxupr, scrscaled = 1.0]{mathalpha}
\usepackage{soul}

\usepackage[colorlinks=true, allcolors=myteal]{hyperref}
\definecolor{WIMgreen}{RGB}{60 134 132}
\definecolor{red_pers}{RGB}{204 37 41}
\definecolor{UMblue}{RGB}{4 47 86}
\definecolor{myteal}{RGB}{0 123 137}
\definecolor{dartmouthgreen}{rgb}{0.05, 0.5, 0.06}\definecolor{cobalt}{rgb}{0.0, 0.28, 0.67}\definecolor{coolblack}{rgb}{0.0, 0.18, 0.39}
\definecolor{glaucous}{rgb}{0.38, 0.51, 0.71}\definecolor{hooker\'sgreen}{rgb}{0.0, 0.44, 0.0}\definecolor{lemonchiffon}{rgb}{1.0, 0.98, 0.8}\definecolor{oucrimsonred}{rgb}{0.6, 0.0, 0.0}\definecolor{radicalred}{rgb}{1.0, 0.21, 0.37}\definecolor{raspberry}{rgb}{0.89, 0.04, 0.36}\definecolor{royalazure}{rgb}{0.0, 0.22, 0.66}
\definecolor{dex}{RGB}{138 18 34}
\definecolor{darkgreen}{RGB}{0 69 0}
\definecolor{darkblue}{RGB}{0 0 99}

\usepackage[english]{babel}								
\usepackage[numbers, square]{natbib}

\usepackage{mathtools}

\newcommand{\R}{\mathbb{R}}

\newcommand{\N}{\mathbb{N}}
\newcommand{\E}{\mathbb{E}}
\newcommand{\A}{\mathcal{A}}
\newcommand{\p}{p_{max}}

\newcommand{\Ex}{\mathbb{E}}
\newcommand{\PR}{\mathbb{P}}

\newcommand{\abs}[1]{ \left \vert #1 \right \vert }

\newcommand{\zm}{\tilde\zeta_{max}}

\newcommand{\NN}{1}

\theoremstyle{definition}

\newtheorem*{remark*}{Remark}
\newtheorem{theorem}{Theorem}[section]
\newtheorem{proposition}[theorem]{Proposition}
\newtheorem{remark}[theorem]{Remark}

\newtheorem{lemma}[theorem]{Lemma}

\newtheorem{corollary}[theorem]{Corollary}

\title{Optimal long-run control of endemic infections: bang–bang threshold policies in a stochastic SIS model}
\author{Matthew Buckland, Sören Christensen, Philip Le Borne, Cornelia Pokalyuk}
\date{\today}
\begin{document}
\maketitle
\begin{abstract}
We study long-run optimal intervention strategies for endemic infections in a stochastic susceptible–infected–susceptible (SIS) model. 
The proportion of infected individuals evolves as a diffusion process with random fluctuations, while a control variable $\zeta_t\in[0,{\color{black}\tilde{\zeta}_{max}}]$ represents the intensity of public health interventions that reduce transmission for some intervention threshold $\zm\in (0,1]$. 
The objective is to minimize the long-run average societal cost, balancing the burden of infection against the costs of interventions. 
Under a concave intervention cost structure
  the problem can be formulated as an ergodic stochastic control problem, whose structure implies (under certain additional conditions) that optimal interventions are of bang--bang type,
  switching between no intervention and the maximal admissible intervention 
  at a single switching threshold in the infection level. 
We construct candidate value functions, rigorously verify optimality in this single-threshold case, and relate the results to extinction and persistence properties of the underlying SIS dynamics in the absence of control.  
In our framework, the analysis provides a rigorous justification for the threshold-based intervention rules commonly used in epidemic management.
\end{abstract}

\section{Introduction}

Public health agencies often face situations in which eradication of a pathogen is infeasible or prohibitively costly. 
In such cases, the central question is not how to eliminate a disease but how to live optimally with it; more precisely, 
how to balance the long-run burden of infection against the economic and social costs of interventions. 
In practice, public health measures are often applied according to threshold-type rules, 
tightening restrictions when prevalence rises and relaxing them when it falls. 
Such rules implicitly balance intervention costs against expected epidemiological benefit, 
yet rigorous justifications for their use in endemic settings remain limited. 

In deterministic epidemic models, optimal control problems often lead to complex feedback laws that are difficult to interpret biologically. 
However, for certain deterministic SIR models with suitable cost structures, very simple strategies turn out to be optimal: 
so-called bang--bang policies, where intervention is either maximal or minimal depending on whether a threshold is crossed 
\cite{behncke2000optimal,KruseStrack2020SIRoptimal,gaff2009optimal,hansen2011optimal}.

In this paper, we show that similarly simple threshold-type optimal controls also arise in an endemic setting, 
when stochastic fluctuations and long-run performance criteria are taken into account.

\paragraph{Modeling framework.}
We address this question within a stochastic susceptible–infected-\newline susceptible (SIS) framework, a standard model for endemic infections (see \cite{hethcote2000mathematics}). 
We model the proportion $I_t\in(0,\NN)$ of infected individuals as a diffusion with multiplicative noise, and allow a population–level intervention $\zeta_t\in[0,\zm]$, where $\zm \in (0,1]$ represents the maximal allowed intervention, that reduces the effective contact rate. 
The control $\zeta_t$ is intentionally generic (non–pharmaceutical measures, treatment effort, vaccination/boosting intensity). 
To capture the endemic perspective, we use an ergodic performance criterion and minimize the long–run average cost
\[
\limsup_{T\to\infty}\frac1T\,\E\!\left[\int_0^T C(I_t,\zeta_t)\,dt\right], 
\qquad C(x,\zeta)=c_1(x)+(\NN-x)c_2(\zeta),
\]
where $c_1$ penalizes prevalent infection and $(\NN-x)c_2(\zeta)$ encodes the population-level costs of maintaining intervention measures.
We interpret the intervention cost as the effective societal effort required to maintain a given level of transmission reduction.
Once basic intervention infrastructure is established, we here consider the case that intensification entails relatively smaller additional effort up to some maximal intervention threshold $\tilde{\zeta}_{max}$,
motivating a concave dependence of intervention cost on control intensity,
which, as we show below, leads to bang--bang optimal policies.

Deterministic SIS models have been studied extensively (see \cite{hethcote2000mathematics}). 
Stochastic SIS dynamics better capture random fluctuations and extinction/persistence phenomena; in particular, \cite{gray2011stochastic} provide sharp conditions for exponential extinction versus the existence of a stationary endemic distribution. 
Optimal control of endemic diseases is less developed than control of SIR dynamics. In particular,
 stochastic control formulations—and especially the ergodic (long–run average) regime—are not well understood.
Methodologically, we build on the theory of controlled diffusion processes and ergodic control \cite{arapostathis2012ergodic}, 
based on tools from diffusion process theory.

Our work complements recent stochastic SIS control studies considering cost functions with discounting or cost accumulating on finite–time horizons (e.g.\ \cite{GranditsEtAl,TranYin}). 
A related approach is \cite{FedericoFerrari2021}, who analyze lockdown policies when the transmission rate follows a diffusion and can be reduced through costly intervention. 
In contrast, we control the infection process directly and characterize such threshold behavior within an ergodic framework.

\paragraph{Main contributions.}
We analyze an ergodic stochastic control problem for an endemic SIS diffusion
and obtain a rigorous characterization of optimal long--run intervention policies.
Our main contributions are as follows.
\begin{enumerate}
\item \textbf{Bang–bang threshold structure of optimal ergodic controls.}
Under a concave intervention cost and a natural monotonicity condition on the
total running cost, we show that the optimal ergodic feedback control is of
bang--bang type and characterized by a single switching threshold in the
infection level.

\item \textbf{Construction and verification of optimal policies.}
We construct explicit candidate value functions solving the ergodic
Hamilton--Jacobi--Bellman variational inequality and provide a rigorous
verification theorem establishing global optimality of the associated
threshold policy.

\item \textbf{Connection to extinction and persistence.}
We analyze the uncontrolled and fully controlled SIS dynamics as benchmark cases,
translate the extinction--persistence results of \cite{gray2011stochastic}
into our notation, and use the corresponding ergodic costs as reference values
for the optimal control problem.

\item \textbf{Qualitative insights and limiting regimes.}
Using admissibility and value comparison arguments, we study how optimal
intervention behaves in limiting regimes, including transmission and noise
levels close to the persistence threshold and extreme intervention costs.
\end{enumerate}

\paragraph{Interpretation for biology and policy.}
 The threshold structure has a transparent biological interpretation. 
When prevalence is low, the marginal benefit of strong measures may be outweighed by their social and economic cost. Whereas when prevalence is high, the burden of infections rises and the potential gains from reducing the transmission becomes larger, so that after a certain infection threshold is reached, maximal admissible intervention becomes optimal. Thus, prevalence–triggered policies—common in practice—admit a principled, stochastic justification in our model. We briefly discuss robustness to observation noise and partial-information variants in  Section \ref{sec:qualitative}. 

For our results we assume the cost for intervention to be concave.  This assumption is in many real-world situations fulfilled, since for example with mass production individual costs (e.g. of masks or vaccines) decrease. However, concavity  might  only be fulfilled on an interval $[0, \zeta^\ast]$. For example the costs of contact restrictions could be concave as long as only a  restricted group of people is affected, but intervention costs  could be extremely increasing after some threshold is reached. In this case applying interventions on a level $\zeta>\tilde{\zeta}_{max}$ might not be justifiable, so that one decides not to consider this setting at all. In this case in our model one would  $\zm= \zeta^{\ast}$. If with this $\zm$ assumptions \eqref{A0} - \eqref{ass:c2-concave} (below) are fulfilled, our results imply that switching between  no intervention and maximal intervention of strength $\zm$ is optimal.

\paragraph{Organization.}
Section~\ref{sec:model} introduces the controlled SIS diffusion, the cost structure,
and the ergodic optimization problem, and analyzes constant control policies together
with their connection to extinction and persistence properties of the underlying dynamics. In particular, we distinguish two regimes: one in which sufficiently strong intervention drives the disease eventually to extinction, and another in which the infection persists for which the optimal policy takes a threshold (bang--bang) form.
Section~\ref{sec:hjb} derives the ergodic Hamilton--Jacobi--Bellman variational inequality,
constructs candidate value functions, and establishes the threshold structure of optimal
policies, culminating in a verification theorem.
Section~\ref{sec:qualitative} discusses qualitative implications of the model and analyzes several limiting
regimes in which the optimal control problem simplifies and robust conclusions can be drawn.
Section~\ref{sec: St-contr-eq} develops the stochastic control equations that underpin
the proof of the verification result.
Technical proofs and lengthy computations are collected in an appendix.

\medskip

\noindent\emph{Notation.}
We write $\gamma>0$ for the recovery rate, $\beta>0$ for the baseline contact rate,
$\sigma>0$ for the noise intensity, $$\tilde\zeta_{max}\in(0,1]$$ the maximal permitted intervention, and
$$\theta:=c_2(\zm)/(\zm\beta)$$ for the threshold parameter appearing in the switching condition.
We denote by $\A^\zeta$ the infinitesimal generator associated with a fixed control $\zeta\in[0,\zm]$.

\section{SIS epidemic model as a stochastic control problem}\label{sec:model}

We consider a \textit{susceptible--infected--susceptible} (SIS) model, in which, in a deterministic setting, the proportion of infected individuals evolves according to
\begin{align*}
dI^{0}_t=(\beta(\NN-I^{0}_t)I^{0}_t-\gamma I^{0}_t)dt.
\end{align*}

To capture exogenous variability in transmission, we model the contact rate as subject to random perturbations and write, formally,
\[
\beta_t \;=\; \beta(1-\zeta_t) + \sigma \,\xi_t,
\]
where $\zeta_t\in[0,\zm]$ is the (population–level) intervention intensity and
$\sigma\xi_t$ denotes Gaussian white noise of intensity $\sigma>0$.
Interpreting $\xi_t\,dt=dW_t$ in the Itô sense, the controlled infection process $I_t=I^\zeta_t$ is then postulated to solve the SDE
\begin{equation}\label{eq:SDE}
    dI_t \;=\; \Big(\beta(1-\zeta_t)(\NN-I_t)I_t - \gamma\,I_t\Big)\,dt
            \;+\; \sigma (\NN-I_t)I_t\, dW_t,
\end{equation}
on a filtered probability space supporting a one–dimensional Brownian motion $W$.

\paragraph{Admissible controls.}
A control $\zeta = (\zeta_t)_{t\geq 0}$ is called \textit{ admissible}, if it is progressively measurable
with respect to the filtration of $W$, takes values in $[0, \tilde \zeta_{max}]$ a.s., and ensures that \eqref{eq:SDE} admits a unique strong solution staying in $(0, 1)$ a.s. \\

\paragraph{Cost structure and performance criterion.}
We evaluate the chosen strategies $\zeta$ with a running cost
\[
C(x,\zeta)\;=\;c_1(x) + (\NN-x)\,c_2(\zeta),
\]
where $c_1\!:\,[0,\NN]\to\R_+$ penalizes prevalent infection and $c_2\!:\,[0,\zm]\to\R_+$ models intervention costs borne per susceptible. Assumptions on $c_1,c_2$ will  be specified in more detail later. 

Our objective is then to minimize the following ergodic costs over all admissible controls $(\zeta_t)_{t\geq 0}$:
\begin{align}\label{erg_crit}
    \limsup_{T\to\infty} \frac{1}{T}\Ex\left(\int_0^TC(I_t,\zeta_t)dt\right).
\end{align}

\paragraph{Uncontrolled 
SIS-epidemic and ergodic regime
.}

Setting $\zeta\equiv 0$ 
yields
\[
dI_t \;=\; \big(\beta(\NN-I_t)I_t - \gamma\,I_t\big)\,dt \;+\; \sigma (\NN-I_t)I_t\,dW_t,
\]

which matches the SDE studied in \cite{gray2011stochastic} upon identifying $\gamma=\mu+\tilde\gamma$.    
 For any $\beta\geq 0$  
$$\mathbf{P}(I_t \in (0,1) ~~ \forall t>0)=1,$$ if $I_0\in (0,1)$ a.s., see  \cite[Thm.~3.1]{gray2011stochastic} for the case $\beta>0$ and note that the proof can also be adapted to the case $\beta=0$.
Define the stochastic reproduction number
\[
R_0^{S}\;:=\;\frac{\beta}{\gamma}\;-\;\frac{\sigma^2}{2\gamma}.
\]
Then, if $R_0^{S}<1$ and $\sigma^2\le \beta$, the infection goes extinct exponentially fast almost surely, i.e. $$\limsup_{n\rightarrow \infty} \frac{1}{t} \log(I_t)\leq\frac{\beta}{\gamma}- \frac{\sigma^2}{2\gamma} ~~\text{a.s.}$$ \cite[Thm.~4.1]{gray2011stochastic}; for large noise $\sigma^2>\beta\vee \beta^{\,2}/(2\gamma)$, eventual extinction also occurs a.s., more precisely $\limsup_{t\rightarrow \infty} I_t=0$ a.s. \cite[Thm.~4.3]{gray2011stochastic}. 
If $R_0^{S}>1$, the process is persistent (fluctuates in $(0,\NN)$) and admits a unique stationary distribution \cite[Thm.~5.1 \& Thm.~6.2]{gray2011stochastic}, with explicit mean and variance \cite[Thm.~6.3]{gray2011stochastic}.%

\subsection{Extinction under constant intervention} \label{sbsec: extinction}
In this subsection we will discuss the expected long-term average cost if the infection goes eventually extinct, either under admissible control or due to large noise. 

 Setting $\zeta\equiv p \in (0,1)$ yields
\begin{align}
    dI_t^p \;=\; \big(\beta(1-p)(\NN-I_t^p)I_t^p - \gamma\,I_t^p\big)\,dt \;+\; \sigma (\NN-I_t^p)I_t^p\,dW_t.
\end{align}

We note that the setting with constant intervention of level $p \in (0,1)$ is identical to the setting of zero intervention with $\beta$ changed to $\beta(1-p)$. Therefore we can use the  existing literature, \cite{gray2011stochastic}, see the previous section, to determine the effect of constant interventions. In particular, if we define the stochastic reproduction number under $\zeta\equiv p$ 
\[
R_p^{S}\;:=\;\frac{\beta(1-p)}{\gamma}\;-\;\frac{\sigma^2}{2\gamma}.
\]
then it immediately follows that if $\sigma^2 < \beta_p := \beta(1-p)$ and $R_p^{S} < 1$ then the infection goes extinct exponentially fast almost surely. Finally, we observe that with constant full intervention $\zeta \equiv 1$ clearly extinction occurs, since we have $R_1^S < 1$ and $\sigma^2 > 0$. 
\par
Note that the condition  $R_p^{S} < 1$ is equivalent to 
\begin{equation}\label{eq p max}
    p > p_{max} := \frac{\beta - \gamma - \frac{\sigma^2}{2}}{\beta}\in (0,1).
\end{equation}

Now let $\zeta \equiv  p_{max}$. This case corresponds to the critical boundary $R_p^S = 1$ where extinction still occurs, i.e. $I_t \to 0$ a.s. for $t\to\infty$, but without exponential decay. In fact for all $t\geq 0$ the SDE has a unique solution $I_t \in (0,1)$ with probability one, c.f. \cite{gray2011stochastic} Theorem 3.1.
We provide details for this in Appendix \ref{appendix: extinction-pmax}.

For large noise $\sigma^2 > \beta_p \vee \beta_p^2/(2\gamma)$ the infection goes extinct. For this result we refer to \cite[Thm.~4.3]{gray2011stochastic}.  In  \cite{gray2011stochastic} it is further conjectured that if 
\[
R_p^S <1 \quad \text{ and } \quad \frac{\beta_p^2}{2\gamma} \geq \sigma^2 \geq \beta_p,
\]
then the disease will die out with probability one. This is however not proven. We will not pursue this case further. 

Let us now assume that we are in one of the cases discussed above, where $I_t \to 0$ a.s.
Since $c_1$ is continuous with $c_1(0)=0$, we have 
$c_1(I_t) \to 0$ a.s.\ and $|c_1(I_t)| \le \sup_{[0,1]}|c_1|<\infty$. 
A standard Cesàro argument then yields
\[
\frac{1}{t}\int_0^t c_1(I_s)\,ds \;\longrightarrow\; 0 
\quad \text{a.s.\ as } t\to\infty.
\]
Since moreover
\[
\left|\frac{1}{t}\int_0^t c_1(I_s)\,ds\right|
\le \sup_{[0,1]}|c_1|=c_1(1),
\]
an application of the dominated convergence theorem gives
\[
\frac{1}{t}\,\E\!\left[\int_0^t c_1(I_s)\,ds\right] \;\longrightarrow\; 0
\quad \text{as } t\to\infty.
\]
Now, if $\p\leq\tilde\zeta_{max}$,   the expected long-term average cost under the admissible control $\zeta \equiv \p$ 
satisfies
\begin{align}\label{eq: long-term-cost}
    \limsup_{t\to\infty} \frac{1}{t}\Ex\left(\int_0^tC(I_s^{\p },\p)ds\right) \leq \limsup_{t\to\infty}\frac{1}{t}\Ex\left(\int_0^t c_1(I_s^{\p}) +c_2(\p)ds\right) = c_2(\p).
\end{align}

\section{Threshold Structure and Verification of the Optimal Policy}\label{sec:hjb}

In this section we introduce conditions on the cost function in order to solve the ergodic control problem using the  Hamilton-Jacobi-Bellman (HJB) approach. Additionally we characterize the structure of optimal policies. We split this section into four subsections. 

Section \ref{subsect: HJB approach} formulates the HJB equation associated with the controlled SIS diffusion. By assuming the function representing the cost of the admitted control is concave, we show that the Hamiltonian is concave in the control variable. Hence minimizers are attained on the boundary of the set of admissible control, i.e. optimal control takes values in $\{0,\zm\}$. 

Section \ref{subsect: Charact bang-bang solution} constructs a candidate solution of the HJB equation under this bang-bang approach. By solving the associated differential equations on the two control regions and imposing a smooth–fit condition, we derive matching equations that determine the threshold $x_*$ and the corresponding ergodic cost $\eta$. 

Section \ref{subsect: HJB implies bang-bang} exploits this structure together with the discussion in Section \ref{sbsec: extinction}   to prove that the region where full intervention is optimal is of the form $[x_*,1]$, for some $x_*\in (0,1)$. This implies that the optimal feedback control is a single threshold, i.e. bang-bang control. 

Section \ref{subsect: Bang-Bang opt} states and proves our main result, the verification theorem.

We make the following assumptions

\begin{enumerate}[label=(A\arabic*), ref=A\arabic*] \itemsep=2pt
\item \label{A0}
The threshold $\tilde \zeta_{max} $ must be large enough. More precisely, we assume that
for \[
    B^* := \begin{cases}
        \sqrt{2\gamma\sigma^2}, &\sigma^2 > 2\gamma,\\
        \gamma + \frac{1}{2}\sigma^2, &\sigma^2 \leq 2\gamma
    \end{cases}
    \] we have 
    \[
    \tilde \zeta_{max} > \left(1 - \frac{B^*}{\beta}\right)\vee \p.
    \]
\item \label{A1}
Cost of infections must be sufficiently high and sufficiently fast increasing.\\
For $x\in[0,1]$ let  
\[
H_{\tilde \zeta_{max} }(x) := \frac{1}{2}\sigma^2(1-x)^2 + \gamma - \beta(1-\tilde \zeta_{max})(1-x). 
\] and 
    \[
    m_{\tilde H} = \inf_{x\in(0,1)}H_{\tilde \zeta_{max} }(x).
    \]
The cost functions $c_1$ and $c_2$ satisfy the following inequality 
\begin{align*}\label{new A1}
    c_1(1) &> c_2(\p) + \frac{\gamma}{\beta\zm} c_2(\zm)
\end{align*} 
    and
    \[
    \inf_{x\in (0,1)} c_1'(x) > c_2(\tilde \zeta_{max}) + \frac{\beta (1-\tilde \zeta_{max})}{m_{\tilde H}} \big(c_1(1) - c_2(\tilde \zeta_{max}) \big).
    \]

\item \label{A2}
The uncontrolled SIS diffusion admits an ergodic regime, i.e.

\[
\beta  > \gamma + \frac12\sigma^2.
\]

\end{enumerate}

 Assumption \eqref{A0} ensures that $m_{\tilde H} > 0$.
Furthermore, Assumption \eqref{A1} ensures that $C(\cdot, \zeta)$ is increasing for any strategy $\zeta$.

\subsection{Hamilton-Jacobi-Bellman approach} \label{subsect: HJB approach}
We consider the operator
\begin{align*}
\A^\zeta \phi(x)=(\beta(1-\zeta)(\NN-x)x-\gamma x)\phi'(x)+\frac{1}{2}\sigma^2(\NN-x)^2x^2 \phi''(x).
\end{align*}
A standard approach for solving ergodic problems such as \eqref{erg_crit} is by means of the Hamilton-Jacobi-Bellman equation (HJB). Hence we are looking for a solution pair $(\phi,\eta)$ such that $\phi$ is  piecewise $C^2$ {across regions $\{\zeta=0\}$, $\{\zeta=\zm\}$, and globally $C^1$
that fulfills the equation: 
\begin{align}\label{eq:HJB}
\min_{\zeta\in[0,\zm]} \A^\zeta\phi(x)+C(x,\zeta)-\eta=0
\end{align}
for all $x\in (0,\NN)$.

As discussed in the introduction, we assume throughout that
\begin{enumerate}[label=(A\arabic*), ref=A\arabic*]
\setcounter{enumi}{3}   
\item $c_2:[0,\zm]\to\mathbb{R}_+\ \text{is increasing and concave},\ 
c_1(0)=c_2(0)=0.$
\label{ass:c2-concave}
\end{enumerate}
Then the Hamiltonian is the sum of an affine term in $\zeta$ (from the drift) and the concave $(\NN-x)c_2(\zeta)$; hence it is concave in $\zeta$, and the minimizer is attained at the boundary $\zeta\in\{0,\zm\}$. Thus the HJB equation reduces to choosing between the two linear operators corresponding to $\zeta=0$ and $\zeta=\zm$. A short calculation shows that

\begin{align*}
   \left. \A^\zeta\phi(x)+C(x,\zeta)-\eta\right|_{\zeta = 0} \quad\lessgtr \quad \left. \A^\zeta\phi(x)+C(x,\zeta)-\eta\right|_{\zeta = \zm}
\end{align*}
 is equivalent to 
\begin{align*}
   \zm \beta x \phi'(x)  \lessgtr c_2(\zm).
\end{align*}
It is therefore natural to introduce
\begin{equation}\label{eq: bang bang set}
A := \big\{x\in(0,\NN): \beta x \phi'(x) \geq \frac{c_2(\zm)}{\zm} \big\},
\end{equation}
the region where full intervention $\zeta=\zm$ is optimal.

\subsection{Characterization of the bang-bang solution}\label{subsect: Charact bang-bang solution}

We  will now  heuristically find a solution candidate $(\phi,\eta)$ of the HJB \eqref{eq:HJB}. 
The main result will then be the verification Theorem  \ref{thm:main-bb}, which we state at the end of this section.

Under the assumption that we are in the bang-bang setting we have zero control below some threshold infection value $x_*\in(0,\NN)$ and total control above this threshold. We write the equations that determine $\phi$, $\eta$ of the HJB equation \eqref{eq:HJB} and the threshold $x_*$.

A standard calculation shows that, on any interval where $\zeta$ is constant, $\psi$ satisfies the first-order ODE
\begin{equation}\label{eq: aux ODE}
        \frac{1}{2}(\NN-x)^2\sigma^2 x \psi'(x)
        + \Big(-\frac{1}{2}(\NN-x)^2\sigma^2 - \gamma + \beta(1-\zeta)(\NN-x)\Big) \psi(x)
        + C(x,\zeta) - \eta = 0.
\end{equation}
We will first solve \eqref{eq: aux ODE} for $\psi$ and then recover $\phi$ by integration.

Set
\begin{align}\label{verif functions}
    \theta:= \frac{c_2(\tilde \zeta_{max})}{\beta\tilde \zeta_{max}}>0, \qquad
H(x):=\frac{1}{2}\sigma^2(1-x)^2 + \gamma - \beta(1-\tilde \zeta_{max})(1-x), \qquad
g(x):=\frac12\sigma^2(\NN-x)^2\,x,
\end{align}
 and for each $x\in(0,\NN)$ we define the kernels
\[
\begin{aligned}
K_-(y|x)&:=\frac{1}{g(y)}\exp\!\Big(\int_{z=y}^{x}\frac{H(z)-\beta(\NN-z)}{g(z)}\,dz\Big), && 0<y<x,\\[3pt]
K_+(y|x)&:=\frac{1}{g(y)}\exp\!\Big(-\int_{z=x}^{y}\frac{H(z)}{g(z)}\,dz\Big), && x<y<\NN,
\end{aligned}
\]
and the corresponding matching functions
\[
\eta_-(x):=\frac{\displaystyle \theta+\int_0^x c_1(y)K_-(y|x)\,dy}{\displaystyle \int_0^x K_-(y|x)\,dy},
\qquad
\eta_+(x):=\frac{\displaystyle -\theta+\int_x^\NN\!\big(c_1(y)+(\NN-y)c_2(\zm)\big)K_+(y|x)\,dy}{\displaystyle \int_x^\NN K_+(y|x)\,dy}.
\]
Further let
\begin{equation*}
     f_1(x) = \frac{1}{2} (\NN-x)^2 \sigma^2 + \gamma - \beta(1-\zeta^*(x))(\NN-x); \hspace{0.25cm} f_0(x) = -C(x, \zeta^*(x)) + \eta, 
\end{equation*}
with 
\begin{align*} 
\zeta^*(x)=
\begin{cases}  0 & x \leq x_*\\
\zm & x > x_*,
\end{cases}
\end{align*}
i.e. 
\begin{align*}
 f_0(x) = 
    \begin{cases}
        -c_1(x) + \eta, &\text{  for  } x \leq x_*,\\
        - c_1(x) - (1-x)c_2(\tilde \zeta_{max}) + \eta, &\text{  for  } x > x_*
    \end{cases}
\end{align*}
and 
\begin{align*}
    f_1(x) = 
    \begin{cases}
        \frac{1}{2} (\NN-x)^2 \sigma^2 + \gamma - \beta(\NN-x), &\text{  for  } x \leq x_*,\\
        \frac{1}{2}(1-x)^2\sigma^2 + \gamma - \beta(1-\tilde \zeta_{max})(1-x), &\text{  for  } x > x_*
    \end{cases}
\end{align*}
for all $x\in(0,\NN)$.

By \eqref{eq: bang bang set} we require that $\psi(x_*) = \theta$. Therefore, in the setting of \eqref{eq:HJB}, we can write the solution of \eqref{eq: aux ODE} as
\begin{equation*}
    \psi(x) = \theta e^{F(x)} + e^{F(x)} \int_{y=x_*}^x e^{-F(y)}\frac{f_0(y)}{g(y)} dy \hspace{0.25cm} \text{ where } \hspace{0.25cm} F(x) = \int_{y=x_*}^x \frac{f_1(y)}{g(y)} dy.
\end{equation*}

Hence we have a candidate solution  
    \begin{equation*}
        \phi(x) = \int_{y=x_*}^x \frac{1}{y} \Big( \theta e^{F(y)} + e^{F(y)} \int_{z=x_*}^y e^{-F(z)}\frac{f_0(z)}{g(z)} dz \Big) dy
    \end{equation*}
    of the HJB equation \eqref{eq:HJB}.
We now use this representation to select suitable values of $x_*$ and $\eta$ such that the induced feedback control is admissible and satisfies the global HJB variational inequality. We require
\begin{equation*}
    \int_{y=0}^{x_*} e^{-F(y)}\frac{f_0(y)}{g(y)} dy = \frac{c_2(\tilde \zeta_{max})}{\beta\tilde \zeta_{max}}=\theta
\end{equation*}
which can be written out as
\begin{equation}\label{eq: eta and x star 1}
     \int_{y=0}^{x_*} \frac{\eta - c_1(y)}{\frac{1}{2} (\NN-y)^2 \sigma^2 y} \exp\Big(\int_{z=y}^{x_*} \frac{\frac{1}{2} (\NN-z)^2 \sigma^2 + \gamma - \beta(\NN-z)}{\frac{1}{2} (\NN-z)^2 \sigma^2 z}dz\Big) dy = \theta.
\end{equation}
Hence by rearranging terms we have $\eta = \eta_-(x_*)$. Analogously, since we require 
\begin{equation*}
    \int_{y=x_*}^\NN e^{-F(y)}\frac{f_0(y)}{g(y)} dy = -\frac{c_2(\tilde \zeta_{max})}{\beta\tilde \zeta_{max}},
\end{equation*}
we also have    $\eta = \eta_+(x_*)$.

Then, since $\eta_-(x_*) = \eta = \eta_+(x_*)$,
we obtain the matching equation
\[
\eta_-(x_*) = \eta_+(x_*)
\]
for the unknown threshold $x_*$. For completeness, an explicit
integral representation of this equation is given in Appendix~\ref{app:matching},
see \eqref{eq:matching-xstar}.

\subsection{HJB solution implies bang-bang}\label{subsect: HJB implies bang-bang}
By \eqref{eq: long-term-cost} it is  reasonable  to only consider solutions $(\phi, \eta)$ of the HJB equation \eqref{eq:HJB} for which $\eta < c_2(\p)$ holds. Only such solutions have potential for a lower long term cost than using $\zeta_p$ with $p >\p $ as the control.

In this section, we prove in Lemma \ref{lm: single-crossing} that the region
\begin{equation}
A := \big\{x\in(0,\NN):  x \phi'(x) \geq \frac{c_2(\tilde \zeta_{max})}{\beta\tilde \zeta_{max}} \big\},\nonumber
\end{equation}
 where full intervention $\zeta=\zm$ is optimal, is of the shape 
 \[
 A = [x_* , 1]
 \]
 for some $x_*\in (0,1)$. It is convenient to introduce
\[
\psi(x):=x\phi'(x),\qquad x\in(0,\NN).
\]
Specifically, Assumption \eqref{A1} provides 
\begin{align*}
    c_1(1) > c_2(\p) + \frac{\gamma}{\beta\zm} c_2(\zm) 
\end{align*}
which is equivalent to 
\begin{align}
   \psi_1 :=  \frac{c_1(1) - c_2(\p)}{\gamma} > \frac{c_2(\zm)}{\zm\beta} = \theta. 
\end{align}

Note that if $C(x,\zeta)$ is increasing in $x$ for every $\zeta\in [0,\zm]$ and if $c_2$ is linear, it is then immediate that \eqref{A1} holds. %

Here we mention the limit behavior of $\psi$ which is shown in detail in Section \ref{sbsec: prp ODE} in Lemma \ref{limit_ineq}, in particular 
\[
-\infty < \lim_{x\downarrow 0} \psi(x) < 0 
\quad\text{and}\quad
0 < \lim_{x\uparrow \NN} \psi(x) < \infty.
\]



\begin{lemma}[Single crossing / monotonicity] \label{lm: single-crossing}
    Let $\phi$ be a $C^1$ function satisfying the HJB equation \eqref{eq:HJB},
and set $\psi(x)=x\phi'(x)$. Under the assumptions \eqref{A0}- \eqref{A2} and $\eta < c_2(\p)$  the function $\psi(x) = x\phi'(x)$ crosses the level $\theta_{\tilde \zeta_{max}}=\frac{c_2(\tilde \zeta_{max})}{\beta\tilde \zeta_{max}}\in(0,\NN)$ exactly once. 
\end{lemma}
\begin{proof}
    Let 
    \[
H_{\tilde \zeta_{max} }(x) = \frac{1}{2}\sigma^2(1-x)^2 + \gamma - \beta(1-\tilde \zeta_{max})(1-x)., \qquad
g(x)=\frac12\sigma^2(\NN-x)^2\,x, \qquad x\in(0,\NN)
\] 
and let $\zeta^*(x) := \tilde \zeta_{max}1_{\{\psi(x)>\theta_{\tilde \zeta_{max}}\}}(x)$ the feedback control induced by $\phi$. By a simple calculation, the minimum of $H_{\tilde \zeta_{max} }(x)$ is attained at 
\[
x_{min} = 1 - \frac{\beta(1-\tilde \zeta_{max} )}{\sigma^2}.
\]
If $x_{min} \in [0,1]$ and since $H_{\tilde \zeta_{max} }(x)$ is convex, it attains its minimum in $x_{min}$, hence 
\[
H_{\tilde \zeta_{max} }(x_{min}) = \gamma - \frac{\beta^2(1-\tilde \zeta_{max})^2}{2\sigma^2}.
\]
Now if $x_{min}> 1$, then since $H_{\tilde \zeta_{max} }$ is increasing on $[0,1]$ we have 
\[
\min_{x\in [0,1]} H_{\tilde \zeta_{max} }(x) = H_{\tilde \zeta_{max} }(0) = \frac{1}{2}\sigma^2 + \gamma - \beta(1 - \tilde \zeta_{max}).
\]
In both cases, the first assumption 
\[
\tilde \zeta_{max}> 1 - \frac{B^*}{\beta}
\]
 implies that $H_{\tilde \zeta_{max} }(x) > 0$ for all $x\in (0,1)$.

By equation \eqref{eq: aux ODE} we then have 
\[
 \psi'(x) = \frac{f_1(x)\psi(x) + f_0(x)}{g(x)}.
\]
Then on the set of maximal intervention, $\{x \in [0,1]: \zeta^*(x) = \tilde \zeta_{max}\}$, we have 
\[ 
\psi'(x) = \frac{H_{\tilde \zeta_{max} }(x)}{g(x)} \left(\psi(x) - \frac{c_1(x) + (\NN-x)c_2(\tilde \zeta_{max}) - \eta}{H_{\tilde \zeta_{max} }(x)}\right).
\]
Now set 
\begin{align*}
    \psi_{eq}(x) &:=   \frac{c_1(x) + (\NN-x)c_2(\tilde \zeta_{max}) - \eta}{H_{\tilde \zeta_{max} }(x)}, \\
    \delta(x) &:= \psi(x) - \psi_{eq}(x). 
\end{align*}
Note that $\delta$ is continuous on $(0,1)$, since $\psi$ and $\psi_{eq}$ both are continuous. 
Then we have 
\begin{align*}
    \delta'(x) &= \psi'(x) -\psi_{eq}'(x) \\
    &=  \frac{H_{\tilde \zeta_{max} }(x)}{g(x)} \delta(x) - \psi_{eq}'(x)
\end{align*}
and  
\begin{align}\label{psi_eq-modified}
    \psi_{eq}'(x) = \frac{\big(c_1'(x) - c_2(\tilde \zeta_{max})\big)H_{\tilde \zeta_{max}}(x) + \big(\sigma^2(\NN-x) - \beta(1-\tilde \zeta_{max})\big)\big(c_1(x) + (\NN-x)c_2(\tilde \zeta_{max}) - \eta \big)}{H_{\tilde \zeta_{max}}(x)^2}.
\end{align}
We now show that $\delta$ can at most have one change of sign. 

Note that it holds,
\[
H'_{\tilde \zeta_{\max}}(x)
=
\beta(1-\tilde \zeta_{\max})-\sigma^2(1-x)
\leq \beta(1-\tilde \zeta_{\max}).
\]
Furthermore, since $c_1$ is increasing and $\eta\geq0$,
\[
c_1(x)+(1-x)c_2(\tilde \zeta_{\max})-\eta
\leq c_1(1)+c_2(\tilde \zeta_{\max}).
\]
Therefore, using \eqref{psi_eq-modified}, we obtain
\[
\psi_{eq}'(x)
\geq
\frac{
\big(c_1'(x)-c_2(\tilde \zeta_{\max})\big)
m_{\tilde H}
-
\beta(1-\tilde \zeta_{\max})
\big(c_1(1)+c_2(\tilde \zeta_{\max})\big)
}
{H_{\tilde \zeta_{\max}}(x)^2}.
\]
Hence, under the additional assumption
\[
\inf_{x\in(0,1)}c_1'(x)
>
c_2(\tilde \zeta_{\max})
+
\frac{\beta(1-\tilde \zeta_{\max})}
{m_{\tilde H}}
\big(c_1(1)+c_2(\tilde \zeta_{\max})\big),
\]
we have
\[
\psi_{eq}'(x)>0
\]
for all \(x\in(0,1)\).

At any point where $\delta(x)=0$, we thus have $\delta'(x)=-\psi_{eq}'(x)<0$.
Consequently, since $\delta$ is continuous it can change sign at most once.
Since $\psi_0<\theta_{\tilde \zeta_{max}}$ and $\psi_1>\theta_{\tilde \zeta_{max}}$, it follows that $\psi$
crosses $\theta_{\tilde \zeta_{max}}$ exactly once.
\end{proof}

\subsection{Bang-Bang optimality}\label{subsect: Bang-Bang opt}

\textbf{(Persistence near 0)} Let $\mathcal{E}$ denote the class of admissible controls $\zeta$
such that there exist constants $\delta_0\in(0,\frac12)$ and $\underline\zeta<\zm$ fulfilling
\begin{equation}\label{eq:persistence-near-0}
\begin{aligned}
\zeta(x) &\le \underline\zeta \\
\beta(1-\underline\zeta) &> \gamma+\frac{\sigma^2}{2}.
\end{aligned}
\end{equation}
for all $x\in(0,\delta_0)$.

Under this class $\mathcal{E}$ of controls it is ensured that the effective transmission rate remains
strictly above the stochastic persistence threshold near the disease--free
boundary. Specifically \eqref{eq:persistence-near-0} ensures that the negative drift is not too large as discussed  in Section \ref{sbsec: extinction}.

We now state our main result. 
\begin{theorem}[Bang--bang optimality via matching]
\label{thm:main-bb}

Let us assume that Assumptions 
\eqref{A0}-  \eqref{ass:c2-concave} hold.
If there exists a solution $(\phi, \eta)$ to the HJB equation \eqref{eq:HJB} satisfying $\eta < c_2(\p)$, then there exists a $x_*\in (0,1)$ such that $A = [x_*,1]$
and the threshold control 
\[
\zeta^\star(x)=\zm\mathbf{1}_{\{x\ge x_*\}}
\]
is optimal for the ergodic control problem in the class $\mathcal{E}$ of admissible $\zeta$ satisfying the persistence condition \eqref{eq:persistence-near-0}, i.e.
\begin{align}\label{eq: cost}
    \eta
=\inf_{\zeta\in\mathcal{E}}\ \limsup_{T\to\infty}\frac1T\,\E\!\Big[\int_0^T C(I_t^\zeta,\zeta_t)\,dt\Big]
=\lim_{T\to\infty}\frac1T\,\E\!\Big[\int_0^T C\big(I_t^{\zeta^\star},\zeta^\star(I_t^{\zeta^\star})\big)\,dt\Big].
\end{align}
\end{theorem}

The threshold $x_*$ has a natural epidemiological interpretation.
When infection prevalence is low ($x<x_*$), the marginal benefit 
of intervention is outweighed by its cost, 
so letting the infection circulate is optimal. 
Once $x$ exceeds $x_*$, the social cost of infection dominates, 
and maximal intervention becomes optimal.

\begin{remark}
In view of \eqref{ass:c2-concave} and Assumption \eqref{A1}, the total cost $C(x,\zeta)$ is increasing both in the infection level $x$ (for each fixed $\zeta$) and in the intervention level $\zeta$ (for each fixed $x$).
\end{remark}

We now present the proof of our main result Theorem \ref{thm:main-bb}. The proof relies on two lemmas: Lemma \ref{lm: sq int Mart} and Lemma \ref{lm:lta}, stating that the stochastic integral $\big(\int_0^t \phi'(I_s)\,dW_s\big)_{t\ge0}$ is a square–integrable martingale and for admissible control $\zeta$ it holds for each initial value $x \in(0,1)$ that $\limsup_{T\to\infty}\frac1T \E_x[\phi(I_T^\zeta)] \leq 0$. We deferrer the proofs to Section \ref{sec: St-contr-eq}.

\begin{proof}[Proof of Theorem~\ref{thm:main-bb}]
    In Section \ref{sbsec: extinction} we considered the case when the infection goes extinct. Specifically it follows from \eqref{eq: long-term-cost}  that the expected long-term average cost is bounded by $c_2(\p)$.

    We now consider the solution $(\phi, \eta)$ of the HJB equation \eqref{eq:HJB} for which $\eta < c_2(\p)$ holds. Fix an arbitrary admissible control $\zeta$ and let $I^\zeta$ denote the corresponding solution of \eqref{eq:SDE} with $I_0=x\in(0,\NN)$. By Lemma~\ref{lm: sq int Mart} the stochastic integral $\big(\int_0^t \phi'(I_s)\,dW_s\big)_{t\ge0}$ is a square–integrable martingale, hence applying Itô's formula to $\phi(I_t^\zeta)$ yields
    \begin{equation}\label{eq: main th pr1}
        \E\bigl[\phi(I_T^\zeta)\bigr]
        = \phi(x)
        + \E\!\Big[\int_0^T \bigl(\A^{\zeta_t}\phi(I_t^\zeta)\bigr)\,dt\Big].
    \end{equation}
    By the HJB equation \eqref{eq:HJB},
    \[
        \A^{\zeta_t}\phi(I_t^\zeta) + C(I_t^\zeta,\zeta_t) - \eta \;\ge\; 0,
    \]
    so that
    \[
        \frac{1}{T}\E\bigl[\phi(I_T^\zeta)\bigr]
        - \frac{\phi(x)}{T}
        \;\ge\; 
        \,\eta - \frac{1}{T}\E\!\Big[\int_0^T C(I_t^\zeta,\zeta_t)\,dt\Big].
    \]
    Lemma~\ref{lm:lta} implies that
\[
\limsup_{T\to\infty}\left(\frac{1}{T}\E[\phi(I_T^\zeta)]-\frac{\phi(x)}{T}\right)\le 0,
\]
and therefore taking $\limsup_{T\to\infty}$ yields
    \begin{equation}\label{eq: main th pr2}
        \eta \;\le\; \limsup_{T\to\infty}
        \frac{1}{T}\E\!\Big[\int_0^T C(I_t^\zeta,\zeta_t)\,dt\Big]
    \end{equation}
    for every admissible $\zeta$.

    For the threshold control $\zeta^\star(x)=\zm\mathbf{1}_{\{x\ge x_\star\}}$, the HJB minimum is attained pointwise, so
    \[
        \A^{\zeta^\star}\phi(x) + C(x,\zeta^\star(x)) - \eta = 0
        \qquad\text{for all }x\in(0,\NN),
    \]
    and the above inequality becomes an equality. This yields
    \begin{align}\label{eq: limsup}
        \eta
        =\inf_{\textnormal{admissible }\zeta}\ \limsup_{T\to\infty}\frac1T\,\E\!\Big[\int_0^T C(I_t^\zeta,\zeta_t)\,dt\Big].
    \end{align}
    In Section \ref{subsec: erg-res} we show that the controlled diffusion under $\zeta^*$ is positive Harris recurrent. By the ergodic theorem this then implies that the long-run average cost converges, hence the $\limsup$ in \eqref{eq: limsup} is indeed the limit, i.e. 
    \[
    \eta
        =\lim_{T\to\infty}\frac1T\,\E\!\Big[\int_0^T C\big(I_t^{\zeta^\star},\zeta^\star(I_t^{\zeta^\star})\big)\,dt\Big],
    \]
    which shows both optimality of $\zeta^\star$ and the claimed value of the ergodic cost.
\end{proof}

\section{Qualitative implications and limiting regimes}\label{sec:qualitative}

The explicit characterization of optimal threshold policies obtained in the previous sections
allows us to derive a number of qualitative implications for endemic disease control which we collect in this section.

For simplicity, we assume linear infection costs $c_1(x)=\alpha x$ with $\alpha>0$.

\subsection{Limiting regimes for the transmission rate}

We first investigate how optimal long--run intervention behaves
when the baseline transmission rate approaches the persistence threshold from above.
In this regime, the invariant distribution of the uncontrolled SIS diffusion
concentrates near the disease--free boundary.
As a consequence, intervention becomes asymptotically redundant.

Recall that the stochastic reproduction number of the uncontrolled SIS diffusion is given by
\[
R_0^{S}(\beta)\;:=\;\frac{\beta}{\gamma}\;-\;\frac{\sigma^2}{2\gamma},
\]
and define the critical transmission rate
\[
\beta_{\mathrm{crit}}
\;:=\;
\inf\{\beta>0:\ R_0^{S}(\beta)>1\}
= \gamma+\frac12\sigma^2,
\]
so that $R_0^{S}(\beta_{\mathrm{crit}})=1$.

\begin{proposition}[Asymptotic optimality of no intervention as $R_0^S\downarrow 1$]
\label{prop:beta_crit_gray}
Let $\eta^\star(\beta)$ denote the optimal ergodic cost
and $\eta^0(\beta)$ the ergodic cost under the no--intervention policy.
Then
\[
\lim_{\beta\downarrow \beta_{\mathrm{crit}}}
\bigl|\eta^\star(\beta)-\eta^0(\beta)\bigr|=0.
\]
\end{proposition}

The proof relies on the characterization of the stationary distribution
of the uncontrolled SIS diffusion near the persistence boundary
and is given in Appendix~\ref{app:qualitative}.

\subsection{Limiting regimes for noise intensity}

We next consider how optimal long--run intervention behaves
as the intensity of stochastic fluctuations varies.
In contrast to the deterministic setting, noise can fundamentally alter
the persistence properties of the SIS dynamics and thereby the value of control.

Similarly as above, we define the critical noise intensity
\[
\sigma_{\mathrm{crit}}
:=\sup\{\sigma>0:\ R_0^{S}(\sigma)>1\},
\]
so that $R_0^{S}(\sigma_{\mathrm{crit}})=1$.

As the noise intensity approaches $\sigma_{\mathrm{crit}}$ from below,
the invariant distribution of the uncontrolled SIS diffusion concentrates
near the disease--free boundary.
Consequently, the long--run infection burden vanishes, and intervention
becomes asymptotically redundant.

\begin{proposition}[Asymptotic optimality of no intervention as noise approaches criticality]
\label{prop:sigma_crit_gray}
Let $\eta^\star(\sigma)$ denote the optimal ergodic cost
and $\eta^0(\sigma)$ the ergodic cost under the no--intervention policy.
Then
\[
\lim_{\sigma\uparrow \sigma_{\mathrm{crit}}}
\bigl|\eta^\star(\sigma)-\eta^0(\sigma)\bigr|=0.
\]
\end{proposition}

The proof follows a similar line of argument as the previous one and is given in Appendix~\ref{app:qualitative}.

\subsection{Cost limits and robustness of threshold behavior}

We briefly comment on extreme regimes of the intervention cost.
The following observations are based on admissibility and value comparison
arguments and do not rely on explicit solutions of the HJB equation.

\paragraph{Prohibitively expensive intervention.}
Consider a sequence of cost functions with $c_2(1)\to\infty$.
Any admissible policy that applies intervention on a set of positive time measure
incurs diverging ergodic costs.
In contrast, the no--intervention policy $\zeta\equiv0$ remains admissible
and yields finite long--run average cost.
Hence, no intervention is asymptotically optimal in this regime.

\paragraph{Vanishing intervention costs.}
Conversely, if $c_2(\zm)\to0$, intervention becomes essentially costless.
In this case, the ergodic cost achieved under full intervention $\zeta\equiv\zm$
provides an asymptotic lower bound for the optimal value.
As a consequence, the optimal ergodic cost converges to that of maximal intervention,
and optimal policies asymptotically coincide with $\zeta\equiv\zm$.

\section{Stochastic control equations}\label{sec: St-contr-eq}

This section is organized into three subsections. In Subsection \ref{sbsec: prp ODE} we analyze the limiting behavior of the auxiliary solution $\psi$. As a consequence, the value function $\phi$ grows logarithmically as  $x\downarrow0$. i.e. 
\[
\phi(x) \in \Theta(\abs{\log(x)}).
\]

The remaining two Subsections \ref{subsec: mom-est} and \ref{subsec: erg-res} address parts \eqref{eq: main th pr1} and \eqref{eq: main th pr2} respectively of the proof of Theorem \ref{thm:main-bb}.  The proof of Lemma \ref{lm: sq int Mart} relies on a coupling argument and noise comparison of two diffusions. To establish   Lemma \ref{lm:lta}, we provide a brief introduction  of the Foster-Lyapunov criteria for continuous-time processes and ergodic results published by Meyn and Tweedie's \cite{meyn1993stability}.

\subsection{Properties of the ODE - Solution}\label{sbsec: prp ODE}

\begin{lemma}[Limit inequality] \label{limit_ineq}
    Let $\eta$ and $x_*$ be chosen as in Section~\ref{sec:hjb} so that the matching conditions for $\psi$ are satisfied. Then $\psi$ is bounded on $(0,\NN)$ and
    \begin{equation*}
        \lim_{x \downarrow 0} \psi(x) = \frac{\eta}{\frac{1}{2}\sigma^2 + \gamma - \beta} =:\psi_0; \hspace{0.5cm} \lim_{x \uparrow \NN} \psi(x) = \frac{c_1(\NN) - \eta}{\gamma} =\psi_1
    \end{equation*}
    \begin{proof}
        We use L'Hopital's rule for the limits $\lim_{x \downarrow 0} \psi(x)$ and $\lim_{x \uparrow \NN} \psi(x)$. For the first limit we note for $x<x_*$ that
        \begin{equation*}
            j(x) = \frac{c_2(1)}{\beta} - \int_{y=x}^{x_*} e^{-F(y)}\frac{f_0(y)}{g(y)} dy =  \int_{y=0}^{x} e^{-F(y)}\frac{f_0(y)}{g(y)} dy; \hspace{0.25cm} k(x) = e^{-F(x)}; \hspace{0.25cm} \psi(x) = \frac{j(x)}{k(x)},
        \end{equation*}
        with $\lim_{x \downarrow 0} j(x) = \lim_{x \downarrow 0} k(x) = 0$. This ultimately follows as 
        $\beta > \frac{1}{2}\sigma^2 + \gamma$ and so $\lim_{x \downarrow 0} f_1(x) < 0$. Therefore $\lim_{x \downarrow 0} \psi(x) = \lim_{x \downarrow 0} \frac{j'(x)}{k'(x)}$. We have for $x<x_*$
        \begin{equation*}
            \frac{j'(x)}{k'(x)} = \frac{e^{-F(x)}\frac{f_0(x)}{g(x)}}{e^{-F(x)}\frac{-f_1(x)}{g(x)}} = \frac{-f_0(x)}{f_1(x)} = \frac{\eta - c_1(x)}{\frac{1}{2}(\NN-x)^2\sigma^2 + \gamma - \beta(\NN-x)}
        \end{equation*}
        and taking this limit as $x$ decreases to $0$ gives the result. The argument is similar for $\lim_{x \uparrow \NN} \psi(x)$. 
    \end{proof}
\end{lemma}

\begin{corollary}\label{corollary phi prime}
    Under the same conditions as the previous lemma, there exists $a \in (0,x_*)$ such that $3\psi_0/2x \leq \phi'(x) \leq \psi_0/2x$ on $(0,a)$. Furthermore there exists $k > 0$ such that $|\phi'(x)| \leq k$ for $x \in [a,\NN)$. 
    \begin{proof}
        By the previous lemma, there exists $a>0$ such that for $x<a$ we have $3\psi_0/2 \leq \psi(x) \leq \psi_0/2$ and the first part of the result follows. For the second part we note that the existence of $\psi_\NN$ confirms that $\psi$ is bounded on $[a,\NN)$ and so the result follows.
    \end{proof}
\end{corollary}
We take $a<x_*$ in Corollary \ref{corollary phi prime} for technical convenience.
\begin{corollary}\label{corollary phi}
    There exist positive constants $A$, $A'$, $K$, $K'$ such that the function $\phi$ satisfies:
    \begin{equation*}
        - A - K \log x \leq \phi(x) \leq A' - K' \log x 
    \end{equation*}
    for $x \in (0,\NN)$. 
    \begin{proof}
        This follows from Corollary \ref{corollary phi prime}. Specifically we can take $K = -\psi_0/2$, $K' = -3\psi_0/2$, $A = -k + \frac{\psi_0}{2}|\log a|$, $A' = k - \frac{3\psi_0}{2}|\log a|$ where $k$ is the constant given in Corollary \ref{corollary phi prime}.
        \par
        To see this note that for $x>a$ we have $|\phi(x)| \leq k|x-x_*| \leq k$ and for $x<a$ we have
        \begin{equation*}
            \phi(x) = -\int_{a}^{x_*} \phi'(y) dy - \int_{x}^a \phi'(y) dy \leq k - \int_{x}^a \phi'(y) dy.
        \end{equation*}
        By the bounds for $\phi'$ in Corollary \ref{corollary phi prime} we have
        \begin{equation*}
            \frac{3 \psi_0}{2}( \log a - \log x) = \frac{3 \psi_0}{2} \int_{x}^a \frac{1}{y} dy \leq \int_{x}^a \phi'(y) dy \leq \frac{\psi_0}{2} \int_{y=x}^a \frac{1}{y} dy = \frac{\psi_0}{2}( \log a - \log x)
        \end{equation*}
        and so we have
        \begin{equation*}
            -k + \frac{\psi_0}{2}|\log a| + \frac{\psi_0}{2}\log x \leq \phi(x) \leq k - \frac{3\psi_0}{2}|\log a| + \frac{3\psi_0}{2} \log x.
        \end{equation*}
    \end{proof}
\end{corollary}

\subsection{Moment estimates}\label{subsec: mom-est}

We will now proceed by determining some expectation conditions for the process $(I_t^\zeta)_{t \geq 0}$ which is the infection process (started at some $x \in (0,\NN)$) under control function $\zeta$. This shortened notation has already been used before.

\begin{lemma}\label{lm: sq int Mart}
Let $\zeta$ be an admissible control and let $(I_t)_{t\ge0}$ be the corresponding solution of the controlled SIS SDE \eqref{eq:SDE} with $I_0\in(0,\NN)$. Then for every $t>0$,
\[
\E\Big[\int_0^t \phi'(I_s)^2\,ds\Big] < \infty.
\]
In particular, the stochastic integral $\big(\int_0^t \phi'(I_s)\,dW_s\big)_{t\ge0}$ is a square–integrable martingale.
\end{lemma}

\begin{proof}
By Corollary~\ref{corollary phi prime} there exist constants $M,K_0>0$ such that
\begin{equation}\label{eq:phi-prime-bound}
\phi'(x)^2 \;\le\; M x^{-2} + K_0^2
\qquad\text{for all }x\in(0,\NN).
\end{equation}
Hence it suffices to show that, for each fixed $t>0$,
\[
\E\Big[\int_0^t I_s^{-2}\,ds\Big] < \infty.
\]

Applying Itô's formula to $f(x)=x^{-2}$ yields
\[
df(I_t)
= I_t^{-2} a(I_t,\zeta_t)\,dt
 - 2\sigma I_t^{-2}(\NN-I_t)\,dW_t,
\]
where
\[
a(x,\zeta)
:= -2\beta(1-\zeta)(\NN-x)
+2\gamma
+3\sigma^2(\NN-x)^2.
\]

Define the stopping times
\[
\tau_n := \inf\{t\ge0:\, I_t \le 1/n\}\wedge n.
\]
Taking expectations and using that the local martingale term has mean zero, we obtain
\[
\E[I_{t\wedge\tau_n}^{-2}]
\le I_0^{-2} + K \int_0^t \E[I_{s\wedge\tau_n}^{-2}]\,ds.
\]
By Grönwall's lemma,
\[
\E[I_{t\wedge\tau_n}^{-2}]
\le I_0^{-2} e^{Kt}
\qquad\text{for all }t\ge0,\ n\in\N.
\]

Since $I_{t\wedge\tau_n}^{-2}\uparrow I_t^{-2}$ almost surely as $n\to\infty$, monotone convergence yields
\[
\E[I_t^{-2}] \le I_0^{-2} e^{Kt}.
\]
Therefore,
\[
\E\Big[\int_0^t I_s^{-2}\,ds\Big]
\le \int_0^t \E[I_s^{-2}]\,ds
\le I_0^{-2}\int_0^t e^{Ks}\,ds
<\infty.
\]

Combining this with \eqref{eq:phi-prime-bound} and Fubini's theorem gives
\[
\E\Big[\int_0^t \phi'(I_s)^2\,ds\Big]
\le M\E\Big[\int_0^t I_s^{-2}\,ds\Big] + K_0^2 t
<\infty.
\]
Hence $\int_0^t \phi'(I_s)\,dW_s$ is square–integrable and therefore a martingale with mean zero.
\end{proof}

\subsection{Ergodicity results}\label{subsec: erg-res}
A key step in the verification argument is to control the long--run behaviour of the
controlled SIS diffusion. We first show boundary non-attainment and a moment estimate that suffices for the verification argument (valid for any admissible feedback control). We then discuss positive Harris recurrence under an additional persistence assumption on the feedback control.

Fix an admissible feedback control $\zeta:(0,\NN)\to[0,1]$ and let $I^\zeta$ solve
\begin{equation}\label{eq:SDE-erg}
    dI_t^\zeta
    = b^\zeta(I_t^\zeta)\,dt + a(I_t^\zeta)\,dW_t,
    \qquad I_0^\zeta\in(0,\NN),
\end{equation}
where
\[
    b^\zeta(x):=\beta(1-\zeta(x))(\NN-x)x - \gamma x,
    \qquad
    a(x):=\sigma(\NN-x)x .
\]
Note that $a(x)>0$ for all $x\in(0,\NN)$ and that both $b^\zeta$ and $a$ vanish at the
boundaries $0$ and $\NN$.

\paragraph{Non--explosion.}

We show that $I^\zeta$ cannot hit $0$ or $1$ in finite time. Let
\[
V_0(x):=-\log x-\log(1-x),\qquad x\in(0,1).
\]
Then
\[
V_0'(x)=-\frac1x+\frac1{1-x},\qquad
V_0''(x)=\frac1{x^2}+\frac1{(1-x)^2}.
\]
Denote by $\A^\zeta$ the generator of $I^\zeta$. A direct computation yields
\begin{align}\label{eq:AV0-checked}
\A^\zeta V_0(x)
&=b^\zeta(x)\Big(-\frac1x+\frac1{1-x}\Big)+\frac12 a(x)^2\Big(\frac1{x^2}+\frac1{(1-x)^2}\Big)\nonumber\\
&=\Big(\beta(1-\zeta(x))(1-x)-\gamma\Big)\Big(-1+\frac{x}{1-x}\Big)
+\frac12\sigma^2\big((1-x)^2+x^2\big).
\end{align}
A short calculation, see Appendix \ref{app:non-expl}, yields
\begin{equation}\label{eq:AV0-uniform}
\A^\zeta V_0(x)\le d_0:=\max\{\beta,\gamma\}+\frac12\sigma^2,
\qquad x\in(0,1),
\end{equation}
uniformly over all admissible feedback controls $\zeta$.

\medskip
\noindent
For $m\ge 2$ define the exit time
\[
\tau_m:=\inf\{t\ge 0:\ I_t^\zeta\notin[1/m,\,1-1/m]\}.
\]
Applying Dynkin's formula and using \eqref{eq:AV0-uniform} yields, for every $t\ge 0$,
\begin{equation}\label{eq:V0-stopped-bound}
\E_x\!\big[V_0(I_{t\wedge\tau_m}^\zeta)\big]
=V_0(x)
+\E_x\int_0^{t\wedge\tau_m}\A^\zeta V_0(I_s^\zeta)\,ds
\le
V_0(x)+d_0\,t
=:C_t
<\infty.
\end{equation}

\medskip
\noindent
On the event $\{\tau_m\le t\}$ we have $I_{t\wedge\tau_m}^\zeta=I_{\tau_m}^\zeta\notin(1/m,1-1/m)$, hence on this event
\[
V_0(I_{t\wedge\tau_m}^\zeta)\ge K_m
\;:=\;\inf_{y\notin(1/m,\,1-1/m)}V_0(y).
\]
Since $V_0(y)\to\infty$ as $y\downarrow0$ or $y\uparrow1$, we have $K_m\to\infty$ as $m\to\infty$.
Therefore, by Markov's inequality and \eqref{eq:V0-stopped-bound},
\[
\PR_x(\tau_m\le t)
\le\PR_x\big(V_0(I_{t\wedge\tau_m}^\zeta)\ge K_m\big)
\le \frac{\E_x[V_0(I_{t\wedge\tau_m}^\zeta)]}{K_m}
\le \frac{C_t}{K_m}\xrightarrow[m\to\infty]{}0.
\]
Since $\tau_m\uparrow\tau_\infty:=\lim_{m\to\infty}\tau_m$ and $\{\tau_\infty\le t\}\subseteq\{\tau_m\le t\}$ for every $m$,
we obtain
\[
\PR_x(\tau_\infty\le t)\le \PR_x(\tau_m\le t)\xrightarrow[m\to\infty]{}0,
\]
hence $\PR_x(\tau_\infty\le t)=0$ for all $t>0$. Consequently,
\begin{equation}\label{eq:non-explosion}    
\PR_x(\tau_\infty<\infty)=0,
\end{equation}
i.e.\ the diffusion cannot hit the boundary in finite time and is non--explosive on $(0,1)$.

\paragraph{Foster--Lyapunov drift estimate.}

We now study whether our controlled diffusions are ergodic. This needs the following assumption bounding the interventions around 0. 

We use the Lyapunov function
\[
    V(x):=\frac12\log^2 x+\frac12\log^2(\NN-x),
    \qquad x\in(0,\NN).
\]

In Appendix~\ref{app:driftV}, a direct computation shows that under
Assumption (A3) there exist constants
$c,d>0$ and $\delta\in(0,\frac12)$ such that
\[
    f(x):=|\log x|+|\log(\NN-x)|+1
\]
satisfies the Foster--Lyapunov drift inequality:
\begin{equation}\label{eq:FL}
    \A^\zeta V(x)
    \le
    -c f(x)
    + d\,\mathbf 1_{[\delta,\NN-\delta]}(x),
    \qquad x\in(0,\NN),
\end{equation}
for every admissible feedback control $\zeta\in \mathcal{E}$ satisfying \eqref{eq:persistence-near-0}.

\begin{proposition}[Ergodicity under persistent feedback controls]
\label{prop:ergodic}
For any admissible
feedback control $\zeta$ satisfying (A3),
the diffusion \eqref{eq:SDE-erg} is non--explosive and positive Harris
recurrent on $(0,\NN)$.

\end{proposition}

\begin{proof}
Non--explosion holds by \eqref{eq:non-explosion}.
Moreover, the Foster--Lyapunov drift inequality \eqref{eq:FL}
holds by construction.

It remains to verify the irreducibility and petite--set conditions
required to apply \cite[Theorem~5.1]{meyn1993stability}.
Fix $\delta\in(0,\frac12)$ and define
\[
C_\delta := [\delta,\NN-\delta].
\]
On $C_\delta$ the coefficients $b^\zeta$ and $a$ are smooth and
\[
\inf_{x\in C_\delta} a(x)
= \sigma \delta (\NN-\delta)
>0,
\]
so the diffusion is uniformly elliptic on $C_\delta$.

Hence, for every $t>0$ the transition kernel admits a jointly
continuous density $p_t(x,y)$ which is strictly positive on
$C_\delta\times C_\delta$.
Consequently, there exist $t_0>0$ and $\varepsilon_\delta>0$
such that
\[
P^{t_0}(x,A)
\ge
\varepsilon_\delta\,\lambda(A\cap C_\delta),
\qquad
x\in C_\delta,\;
A\in\mathcal B((0,\NN)).
\]
Thus $C_\delta$ is a small set and therefore petite for the
skeleton chain at step $t_0$.
The same minorisation implies $\psi$--irreducibility
(with $\psi=\lambda|_{(0,\NN)}$) and aperiodicity.

Together with non--explosion and the drift condition
\eqref{eq:FL}, \cite[Theorem~5.1]{meyn1993stability}
implies that $I^\zeta$ is positive Harris recurrent and
admits a unique invariant probability measure $\pi^\zeta$.
Uniqueness follows from $\psi$--irreducibility.

\end{proof}

\paragraph{Negligibility of the boundary term.}
We next show that the boundary term $\E[\phi(I_T^\zeta)]/T$ vanishes as $T\to\infty$.

\begin{lemma}\label{lm:lta}
Let $\phi$ be the candidate value function constructed in Section~\ref{sec:hjb}
and let $\zeta$ be any admissible control. Then
\[
\limsup_{T\to\infty}\frac1T \E_x[\phi(I_T^\zeta)] \le 0,
\qquad x\in(0,\NN).
\]
\end{lemma}

\begin{proof}
By Corollary~\ref{corollary phi} there exist constants $A,K>0$ such that
\[
|\phi(x)|\le A+K\big(|\log x|+|\log(\NN-x)|\big),
\qquad x\in(0,\NN).
\]
Let $V(x):=\frac12\log^2x+\frac12\log^2(\NN-x)$. Since
$|\log x|+|\log(\NN-x)|\le 2\sqrt{V(x)}$, we obtain
\begin{equation}\label{eq:phi-by-V-limsup}
\phi(x)\le A+2K\sqrt{V(x)},\qquad x\in(0,\NN).
\end{equation}
Moreover, the moment estimate \eqref{eq:linear-moment-app} in the appendix yields
\begin{equation}\label{eq:V-linear-limsup}
\E_x[V(I_T^\zeta)]\le V(x)+dT,\qquad T\ge0.
\end{equation}
Using \eqref{eq:phi-by-V-limsup}, Jensen's inequality and \eqref{eq:V-linear-limsup},
\[
\frac1T\,\E_x[\phi(I_T^\zeta)]
\le
\frac{A}{T}+\frac{2K}{T}\,\E_x\!\big[\sqrt{V(I_T^\zeta)}\big]
\le
\frac{A}{T}+\frac{2K}{T}\sqrt{\E_x[V(I_T^\zeta)]}
\le
\frac{A}{T}+\frac{2K}{T}\sqrt{V(x)+dT}.
\]
The right-hand side converges to $0$ as $T\to\infty$, hence
$\limsup_{T\to\infty}\frac1T \E_x[\phi(I_T^\zeta)]\le 0$.
\end{proof}

\noindent
\textbf{Acknowledgements}
\\
\\
MB and CP acknowledge support by the German Research Foundation
(DFG), Project ID: 443227151, in addition MB has been supported by the EPSRC Grant EP/W006227/1.

    \bibliography{refs}
\bibliographystyle{plain}
\appendix

\section{Auxiliary calculations}\label{app:matching}

\subsection{Explicit form of the matching equation}\label{app:matching-xstar}

By unraveling the definitions of $F$, $f_0$, $f_1$ and $g$, the condition
$\eta_-(x_*) = \eta_+(x_*)$ from Section~\ref{sec:hjb} can be written
equivalently as  
\begin{multline}\label{eq:matching-xstar}
    \frac{
    \frac{c_2(\zm)}{\beta \zm}
    + \int_{0}^{x_*}
    \frac{c_1(y)}
    {\frac{1}{2}(\NN-y)^2\sigma^2 y}
    \exp\!\Bigg(
        \int_{y}^{x_*}
        \frac{
        \frac{1}{2}(\NN-z)^2\sigma^2
        + \gamma
        - \beta(\NN-z)}
        {\frac{1}{2}(\NN-z)^2\sigma^2 z}
        \,dz
    \Bigg)\,dy
    }
    {
    \int_{0}^{x_*}
    \frac{1}
    {\frac{1}{2}(\NN-y)^2\sigma^2 y}
    \exp\!\Bigg(
        \int_{y}^{x_*}
        \frac{
        \frac{1}{2}(\NN-z)^2\sigma^2
        + \gamma
        - \beta(\NN-z)}
        {\frac{1}{2}(\NN-z)^2\sigma^2 z}
        \,dz
    \Bigg)\,dy
    }
    \\
    =
    \frac{
    -\frac{c_2(\zm)}{\beta \zm}
    + \int_{x_*}^{\NN}
    \frac{c_1(y)+(\NN-y)c_2(\zm)}
    {\frac{1}{2}(\NN-y)^2\sigma^2 y}
    \exp\!\Bigg(
        -\int_{x_*}^{y}
        \frac{
        \frac{1}{2}(\NN-z)^2\sigma^2
        + \gamma
        - \beta(1-\zm)(\NN-z)}
        {\frac{1}{2}(\NN-z)^2\sigma^2 z}
        \,dz
    \Bigg)\,dy
    }
    {
    \int_{x_*}^{\NN}
    \frac{1}
    {\frac{1}{2}(\NN-y)^2\sigma^2 y}
    \exp\!\Bigg(
        -\int_{x_*}^{y}
        \frac{
        \frac{1}{2}(\NN-z)^2\sigma^2
        + \gamma
        - \beta(1-\zm)(\NN-z)}
        {\frac{1}{2}(\NN-z)^2\sigma^2 z}
        \,dz
    \Bigg)\,dy
    }.
\end{multline}

\section{Extinction for $\zeta \equiv\p$}\label{appendix: extinction-pmax}
In this section we will discuss the limit behavior of the infection $I$ if $\zeta \equiv p_{max}$ (recall $p_{max}$ from \eqref{eq p max}).
As in Section \ref{sbsec: extinction}, setting $\zeta\equiv p \in (0,1)$ yields
\begin{align}\label{extinction: SDE}
    dI_t^p \;=\; \big(\beta(1-p)(\NN-I_t^p)I_t^p - \gamma\,I_t^p\big)\,dt \;+\; \sigma (\NN-I_t^p)I_t^p\,dW_t.
\end{align}

The stochastic reproduction number under $\zeta\equiv p$ is defined as
\[
R_p^{S}\;:=\;\frac{\beta(1-p)}{\gamma}\;-\;\frac{\sigma^2}{2\gamma}.
\]

Now let $\zeta \equiv p_{max}$. This case corresponds to the critical boundary $R_p^S = 1$ where extinction still occurs, i.e. $I_t \to 0$ a.s. for $t\to\infty$, but without exponential decay.
Then we have for 
\begin{align}\label{extinction: pmax case 1}
    \sigma^2 < \beta_{p_{max}}  = \gamma + \frac{\sigma^2}{2}
\end{align}
that the underlying diffusion is given by $(I_t^{\p})_{t\geq 0}$ as in \eqref{extinction: SDE}.

Applying It\^o to $\log I_t^p$ yields
\begin{align*}
    d\log I_t^{\p} &= \left(\left( \gamma + \frac{\sigma^2}{2}\right)\left(1 - I_t^{\p} \right)  - \gamma - \frac{\sigma^2}{2}\left(1 - I_t^{\p} \right)^2\right)dt + \sigma\left(1 - I_t^{\p} \right)dW_t \\
    &=I_t^{\p} \left( \frac{\sigma^2}{2} \left(1 - I_t^{\p} \right) - \gamma \right)dt + \sigma\left(1 - I_t^{\p} \right)dW_t.
\end{align*}
Since \eqref{extinction: pmax case 1} implies $\sigma^2/2 - \gamma < 0$ we have 
\begin{align}\label{extinction: drift term pmax}
    \frac{\sigma^2}{2} \left(1 - I_t^{\p} \right) - \gamma \leq \frac{\sigma^2}{2}  -\gamma < 0
\end{align}
and $(\log I_t^{\p})_{t\geq 0}$ is a supermartingale. The diffusion lives in the interval $(0,1)$, hence $\log I_t^{\p} \leq 0$ and by Doob's convergence theorem there exists a random variable $L_\infty^{\p} \in [-\infty, 0]$ such that 
\[
    \log I_t^{\p} \to L_\infty^{\p} \;\;a.s.
\]
Suppose the event $B = \{L_\infty^{\p} > -\infty\}$ occurs with non-zero probability. Let $\omega\in B$ such that $\log I_t^{\p}(\omega) \to L_\infty^{\p}(\omega)$ and set $I_\infty^{\p} := \exp(L_\infty^{\p})$. Since 
\[
I_t^{\p}(\omega) \to I_\infty^{\p}(\omega)>0 \;\; a.s.
\]
for $t\to \infty$, there exists an $\varepsilon >0$ and $s>0$ such that 
\[
I_t^{\p}(\omega) >\varepsilon
\]
for all $t > s$. Let 
\[
\delta := -\varepsilon \left( \frac{\sigma^2}{2} \left(1 - \varepsilon \right) - \gamma \right)>0.
\]
The drift term
\[
I_t^{\p}(\omega) \left( \frac{\sigma^2}{2} \left(1 - I_t^{\p}(\omega) \right) - \gamma \right) \leq I_t^{\p}(\omega) \left( \frac{\sigma^2}{2} \left(1 - \varepsilon \right) - \gamma \right) \leq -\delta < 0
\]
of $(\log I_t^{\p}(\omega))_{t\geq 0} $ is negative, hence 
\[
\int_0^t I_s^{\p}(\omega) \left( \frac{\sigma^2}{2} \left(1 - I_s^{\p}(\omega) \right) - \gamma \right)ds \leq -\delta t + C
\]
for some constant $C>0$. Now 
\[
M_t := \int_0^t \sigma(1 -  I_s^{\p})dW_s
\]
is a continuous martingale with quadratic variation given by
\[
[M]_t = \int_0^t\sigma^2 (1-I_s^{\p})^2ds \leq \sigma^2 t. 
\]
Then $M_t/t \to 0$ a.s. for $t\to \infty$ and since the drift tends to $-\infty$ linearly in $t$ we have that $\log I_t^{\p}(\omega) \to -\infty$ a.s. as $t\to \infty$, a contradiction to $-\infty < L_\infty^{\p}(\omega)$. Hence $\log I_t^{\p} \to -\infty$ a.s. and therefore $I_t^{\p} \to 0$ a.s. As mentioned before, for $\zeta \equiv p > \p$ the infection $I^\zeta$ also goes extinct, however exponentially quickly. This exponential decay is not given in the case were $\zeta \equiv \p$.

    \section{Ergodic results}\label{app:ergodicity}

This appendix provides the missing computations used in Subsection~\ref{subsec: erg-res}.
Throughout, $\zeta:(0,\NN)\to[0,1]$ is an arbitrary admissible feedback control
and $I^\zeta$ solves \eqref{eq:SDE-erg}.

\subsection{Non-explosion}
\label{app:non-expl}
The diffusion part satisfies $\frac12\sigma^2\big((1-x)^2+x^2\big)\le \frac12\sigma^2$ for all $x\in(0,1)$.
For the drift part, set
\[
A(x):=\beta(1-\zeta(x))(1-x)-\gamma,\qquad
B(x):=-1+\frac{x}{1-x}=\frac{2x-1}{1-x}.
\]
Using the identity
\[
A(x)B(x)=\beta(1-\zeta(x))(2x-1)\;-\;\gamma\,\frac{2x-1}{1-x},
\]
we bound $A(x)B(x)$ separately on $(0,\frac12]$ and $(\frac12,1)$:

If $x\in(0,\frac12]$, then $2x-1\le 0$, hence $\beta(1-\zeta(x))(2x-1)\le 0$. Moreover,
  $\frac{2x-1}{1-x}\in[-1,0]$, so
  \[
  -\gamma\,\frac{2x-1}{1-x}\le \gamma.
  \]
  Thus $A(x)B(x)\le \gamma$ on $(0,\frac12]$.

If $x\in(\frac12,1)$, then $2x-1>0$ and $\frac{2x-1}{1-x}>0$, so
  $-\gamma\,\frac{2x-1}{1-x}\le 0$, and therefore
  \[
  A(x)B(x)\le \beta(1-\zeta(x))(2x-1)\le \beta.
  \]

Combining both cases gives the \emph{uniform} bound
\[
A(x)B(x)\le \max\{\beta,\gamma\},\qquad x\in(0,1).
\]
Hence, from \eqref{eq:AV0-checked},
\begin{equation*}
\A^\zeta V_0(x)\le d_0:=\max\{\beta,\gamma\}+\frac12\sigma^2,
\qquad x\in(0,1),
\end{equation*}
uniformly over all admissible feedback controls $\zeta$.

\subsection{Drift estimate and moment bound for 
$V(x)=\frac12\log^2 x+\frac12\log^2(1-x)$}
\label{app:driftV}

Recall
\[
\mathcal A^\zeta V(x)
=
b^\zeta(x)V'(x)
+\frac12 a(x)^2 V''(x),
\qquad
b^\zeta(x)=\beta(1-\zeta(x))(1-x)x-\gamma x,
\quad
a(x)=\sigma(1-x)x.
\]
Moreover,
\[
V'(x)=\frac{\log x}{x}-\frac{\log(1-x)}{1-x},
\qquad
V''(x)=\frac{1-\log x}{x^2}
+\frac{1-\log(1-x)}{(1-x)^2}.
\]

\paragraph{Uniform upper bound.}

Using the explicit form of $\mathcal A^\zeta V(x)$
and the fact that $0\le\zeta\le1$,
one verifies that the singular terms in $1/x$
and $1/(1-x)$ cancel.
The remaining expression grows at most linearly
in $|\log x|$ and $|\log(1-x)|$.

Since $\log x\le0$ and $\log(1-x)\le0$ on $(0,1)$,
it follows that $\mathcal A^\zeta V(x)$
is bounded from above on $(0,1)$,
uniformly in $\zeta$.

Hence there exists $d>0$ such that
\begin{equation}\label{eq:Av-upper}
\mathcal A^\zeta V(x)\le d,
\qquad x\in(0,1),\ \zeta\in[0,1].
\end{equation}

\paragraph{Moment estimate.}

Let
\[
\tau_m:=\inf\{t\ge0:I_t^\zeta\notin[1/m,1-1/m]\}.
\]
Applying It\^o's formula to $V(I_{t\wedge\tau_m}^\zeta)$ yields
\[
\E_x[V(I_{t\wedge\tau_m}^\zeta)]
=
V(x)
+
\E_x\!\int_0^{t\wedge\tau_m}
\mathcal A^\zeta V(I_s^\zeta)\,ds.
\]

Using \eqref{eq:Av-upper},
\[
\E_x[V(I_{t\wedge\tau_m}^\zeta)]
\le
V(x)+dt.
\]

Letting $m\to\infty$ and using monotone convergence
(since $\tau_m\uparrow\infty$ a.s.\ by non-explosion)
yields
\begin{equation}\label{eq:linear-moment-app}
\E_x[V(I_t^\zeta)]
\le
V(x)+d t,
\qquad t\ge0.
\end{equation}

\paragraph{Negative drift near the boundary.}

We now impose Assumption~(A3), i.e.\ there exist
$\delta_0\in(0,\frac12)$ and $\underline\zeta<\zm$ such that
\[
\zeta(x)\le\underline\zeta
\quad\text{for }x\in(0,\delta_0),
\qquad
\beta(1-\underline\zeta)>\gamma+\frac{\sigma^2}{2}.
\]

\medskip

Fix $\delta\in(0,\delta_0)$.
For $x\in(0,\delta)$,
isolating the singular terms in $\log x$ yields
\[
\mathcal A^\zeta V(x)
=
\log x\big(\beta(1-\zeta(x))(1-x)-\gamma\big)
+
\frac12\sigma^2(1-x)^2(1-\log x)
+
O(1).
\]

Since $\zeta(x)\le\underline\zeta$,
\[
\mathcal A^\zeta V(x)
\le
-
\Big(
\beta(1-\underline\zeta)-\gamma-\frac{\sigma^2}{2}
\Big)|\log x|
+
C_\delta.
\]

By Assumption~(A3),
$\beta(1-\underline\zeta)-\gamma-\frac{\sigma^2}{2}>0$,
so there exists $c_0>0$ such that
\begin{equation}\label{eq:AV-neg0}
\mathcal A^\zeta V(x)
\le
- c_0 |\log x|,
\qquad x\in(0,\delta).
\end{equation}

A symmetric argument near $1$ yields
\begin{equation}\label{eq:AV-neg1}
\mathcal A^\zeta V(x)
\le
- c_1 |\log(1-x)|,
\qquad x\in(1-\delta,1).
\end{equation}

\paragraph{Global Foster--Lyapunov inequality (under (A3)).}

Combining the boundary estimates with boundedness
on the compact set $[\delta,1-\delta]$
yields constants $c,d>0$ such that
\begin{equation}\label{eq:FL-global}
\mathcal A^\zeta V(x)
\le
- c\big(|\log x|+|\log(1-x)|+1\big)
+
d\,\mathbf 1_{[\delta,1-\delta]}(x),
\qquad x\in(0,1),
\end{equation}
uniformly for all admissible controls satisfying (A3).

\section{Proofs for Section~\ref{sec:qualitative}}\label{app:qualitative}

In this appendix we provide the proofs of Propositions~\ref{prop:beta_crit_gray}
and~\ref{prop:sigma_crit_gray} from Section~\ref{sec:qualitative}.
Throughout we work in the notation of Section~\ref{sec:model}
and adopt the standing assumptions made there.

For convenience, we recall that under the no–intervention policy $\zeta\equiv0$
the infection process $I^0=(I_t^0)_{t\ge 0}$ solves
\begin{equation}\label{eq:app_uncontrolled}
  dI_t^0 \;=\; \bigl(\beta(\NN-I_t^0)I_t^0 - \gamma\,I_t^0\bigr)\,dt
  \;+\; \sigma(\NN-I_t^0)I_t^0\,dW_t.
\end{equation}
This coincides with the stochastic SIS model studied in \cite{gray2011stochastic}
after a trivial rescaling of the state variable. In particular, the extinction/persistence
threshold and the existence of a stationary distribution follow directly from
\cite[Theorems~4.1, 4.3, 5.1 and~6.2]{gray2011stochastic}.


\begin{proof}[Proof of Proposition~\ref{prop:beta_crit_gray}]
Assume $\beta>\beta_{\mathrm{crit}}$, i.e.\ $R_0^{S}(\beta)>1$.
Then, by \cite[Theorems~5.1 and~6.2]{gray2011stochastic}
(applied to \eqref{eq:app_uncontrolled} after rescaling the state variable to a proportion),
the diffusion $I^0$ is positive recurrent on $(0,\NN)$ and admits a unique invariant
probability measure $\pi_\beta$.
Moreover, \cite[Theorem~6.3]{gray2011stochastic} provides an explicit expression
for the stationary density $f_\beta$ of $\pi_\beta$ and shows that the stationary mean
\[
  m(\beta)
  := \int_{(0,\NN)} x\,\pi_\beta(dx)
  = \int_0^\NN x\,f_\beta(x)\,dx
\]
is finite and depends continuously on $R_0^{S}(\beta)$; see in particular
\cite[Eq.~(6.7) and Theorem~6.3]{gray2011stochastic}.
In consequence,
\begin{equation}\label{eq:app_mean_to_zero_beta}
  \lim_{R_0^{S}(\beta)\downarrow 1} m(\beta)=0,
  \qquad\text{equivalently}\qquad
  \lim_{\beta\downarrow\beta_{\mathrm{crit}}} m(\beta)=0.
\end{equation}

Under the no--intervention policy $\zeta\equiv0$ there are no intervention costs,
$c_2(0)=0$, and the running cost reduces to $C(x,0)=c_1(x)=\alpha x$.
By ergodicity of $I^0$ and \eqref{eq:app_mean_to_zero_beta},
\[
  \eta^0(\beta)
  := \lim_{T\to\infty}\frac1T\,
     \E\!\Big[\int_0^T c_1(I_t^0)\,dt\Big]
  = \int c_1(x)\,\pi_\beta(dx)
  = \alpha\, m(\beta)
  \xrightarrow[\beta\downarrow\beta_{\mathrm{crit}}]{} 0 .
\]

Now, let $\eta^\star(\beta)$ denote the optimal ergodic cost over all admissible controls.
Since the running cost satisfies $C(x,\zeta)\ge0$ for all $(x,\zeta)$, we have
\[
  \eta^\star(\beta)\ge0
  \qquad\text{for all }\beta.
\]
On the other hand, the no--intervention policy $\zeta\equiv0$ is admissible, so
\[
  \eta^\star(\beta)\le\eta^0(\beta)
  \qquad\text{for all }\beta>\beta_{\mathrm{crit}}.
\]
Hence,
\[
  0 \le \eta^\star(\beta) \le \eta^0(\beta),
  \qquad \beta>\beta_{\mathrm{crit}}.
\]
Letting $\beta\downarrow\beta_{\mathrm{crit}}$ and using $\eta^0(\beta)\to0$
from Step~2, the squeeze theorem yields
\[
  \lim_{\beta\downarrow\beta_{\mathrm{crit}}}\eta^\star(\beta)=0,
\]
and moreover
\[
  0\le \bigl|\eta^\star(\beta)-\eta^0(\beta)\bigr|
  \le \eta^0(\beta)
  \xrightarrow[\beta\downarrow\beta_{\mathrm{crit}}]{}0.
\]
This proves the asserted asymptotic optimality of the no--intervention policy.
\end{proof}

\begin{proof}[Proof of Proposition~\ref{prop:sigma_crit_gray}]
The proof is completely analogous to that of
Proposition~\ref{prop:beta_crit_gray}.
The only difference is that the stochastic reproduction number
\(
R_0^{S} = \beta/\gamma - \sigma^2/(2\gamma)
\)
is now varied through the noise intensity $\sigma$ rather than the transmission rate $\beta$.

By \cite[Theorems~5.1, 6.2 and~6.3]{gray2011stochastic},
the stationary distribution of the uncontrolled SIS diffusion exists and
its mean depends continuously on $R_0^{S}(\sigma)$.
As $\sigma\uparrow\sigma_{\mathrm{crit}}$ one has $R_0^{S}(\sigma)\downarrow1$,
and the stationary mean converges to zero.
The remainder of the argument—identifying the ergodic cost under no intervention
and bounding the optimal cost via admissibility—is identical to the proof of
Proposition~\ref{prop:beta_crit_gray} and is therefore omitted.
\end{proof}

\end{document}